\documentclass[twocolumn,amsthm]{autartArXiv}

\newcounter{saveenumerate}
\makeatletter
\newcommand{\assumptionintertext}[1]{%
\setcounter{saveenumerate}{\value{enum\romannumeral\the\@enumdepth}}
\end{enumerate}
#1
\begin{enumerate}
\renewcommand\labelenumi{(A\theenumi)}
\setcounter{enum\romannumeral\the\@enumdepth}{\value{saveenumerate}}%
}
\makeatother

\usepackage{graphicx}
\graphicspath{{./figures/}}
\DeclareGraphicsExtensions{.pdf,.jpg,.png}

\usepackage{amsmath}
\usepackage{amssymb}
\DeclareMathOperator*{\argmin}{arg\,min}

\usepackage{mathtools}
\DeclarePairedDelimiter{\floor}{\lfloor}{\rfloor}

\usepackage{siunitx}

\usepackage{amsthm}
\newtheorem{notation}{Notation}
\newtheorem{definition}{Definition}
\newtheorem{theorem}{Theorem}
\newtheorem{remark}{Remark}

\usepackage{dsfont}

\usepackage[round,longnamesfirst]{natbib}

\usepackage[colorlinks=true,citecolor=blue]{hyperref}

\usepackage{booktabs}

\begin{document}

\begin{frontmatter}

\title{Initial estimates for Wiener-Hammerstein models using phase-coupled multisines\thanksref{footnoteinfo}}

\thanks[footnoteinfo]{The material in this paper was partially presented at the 19th IFAC World Congress, August 24--29, 2014, Cape Town, South Africa. Corresponding author K.~Tiels. Tel. +32 2 6293665. 
Fax +32 2 6292850.}

\author[Brussels]{Koen Tiels}\ead{Koen.Tiels@vub.ac.be},
\author[Brussels]{Maarten Schoukens}\ead{Maarten.Schoukens@vub.ac.be},
\author[Brussels]{Johan Schoukens}\ead{Johan.Schoukens@vub.ac.be}

\address[Brussels]{Department ELEC, Vrije Universiteit Brussel, 1050 Brussels, Belgium}


\begin{abstract}
Block-oriented models are often used to model nonlinear systems. These models consist of linear dynamic (L) and nonlinear static (N) sub-blocks. This paper addresses the generation of initial estimates for a Wiener-Hammerstein model (LNL cascade). While it is easy to measure the product of the two linear blocks using a Gaussian excitation and linear identification methods, it is difficult to split the global dynamics over the individual blocks. This paper first proposes a well-designed multisine excitation with pairwise coupled random phases. Next, a modified best linear approximation is estimated on a shifted frequency grid. It is shown that this procedure creates a shift of the input dynamics with a known frequency offset, while the output dynamics do not shift. The resulting transfer function, which has complex coefficients due to the frequency shift, is estimated with a modified frequency domain estimation method. The identified poles and zeros can be assigned to either the input or output dynamics. Once this is done, it is shown in the literature that the remaining initialization problem can be solved much easier than the original one. The method is illustrated on experimental data obtained from the Wiener-Hammerstein benchmark system.
\end{abstract}

\end{frontmatter}

\section{Introduction}
\label{sec: Introduction}

Even if all physical dynamic systems behave nonlinearly to some extent, we often use linear models to describe them. If the nonlinear distortions get too large, a linear model is insufficient, and a nonlinear model is required.

One possibility is to use block-oriented models \citep{Billings1982,Giri2010}, which combine linear dynamic (L) and nonlinear static (N), i.e. memoryless, blocks.
Due to this highly structured nature, block-oriented models offer insight about the system to the user. This can be useful in e.g. fault detection, to detect in which part of the system a fault occurred, e.g. changing dynamics in only part of the model.
Block-oriented models are preferred when there are localized nonlinearities in the system, thus leading to a sparse representation of the system in terms of interconnected blocks.
Due to the separation between the dynamics and the nonlinearities, block-oriented models also allow for an easy discretization (i.e. the conversion from a continuous-time to a discrete-time representation). We refer the reader to \citet{Giri2010} for an elaborated discussion.
The simplest block-oriented models are the Wiener model (LN cascade) and the Hammerstein model (NL cascade). They can be generalized to a Wiener-Hammerstein model (LNL cascade, see \autoref{fig: Wiener-Hammerstein system}).
Applications of Wiener-Hammerstein models can mainly be found in biology \citep{Korenberg1986,Dewhirst2010,Bai2009}, but also in the modeling of RF power amplifiers \citep{Isaksson2006}.

Several identification methods have been proposed to identify Wiener-Hammerstein systems.
Early work can be found in \citet{Billings1982,Korenberg1986}. The maximum likelihood estimate is formulated in \citet{Chen1992}. Wiener-Hammerstein systems are modeled as the cascade of well-selected Hammerstein models in \citet{Wills2012}. The recursive identification of error-in-variables Wiener-Hammerstein systems is considered in \citet{Mu2014}.
Both \citet{Chen1992} and \mbox{\citet{Wills2012}} indicate the importance of good initial estimates, but not how to obtain them.
\citet{Sjoberg2012b} indicates the importance of good initial estimates on an example. The optimization of the model parameters can either converge extremely slowly or get trapped in a local optimum, even if the correct number of poles and zeros is assigned to both the input and the output dynamics, leading to Wiener-Hammerstein models that only fit about as well as a linear model.

Some approaches obtain initial estimates by using specifically designed experiments. For example, \citet{Vandersteen1997} proposes a series of experiments with large and small signal multisines. \citet{Weiss1998} uses only two experiments with paired multisines, but the approach requires the estimation of the Volterra kernels of the system. \citet{Crama2005} proposes an iterative initialization scheme that only requires one experiment of a well-designed multisine excitation.

A major difficulty is the generation of good initial values for the two linear blocks $R(q)$ and $S(q)$ of the Wiener-Hammerstein system (see \autoref{fig: Wiener-Hammerstein system}).
An initial estimate for the static nonlinearity can be obtained using a simple linear regression if a basis function expansion, linear-in-the-parameters, for the nonlinearity is used, and if the dynamics are initialized.
The poles and the zeros of both $R(q)$ and $S(q)$ can be obtained from the best linear approximation (\textsc{BLA}) \citep{Pintelon2012} of the Wiener-Hammerstein system. To obtain initial estimates for $R(q)$ and $S(q)$, the poles and the zeros of the \textsc{BLA} should be split over the individual transfer functions $R(q)$ and $S(q)$. Several methods have been proposed to make this split.
The brute-force method in \citet{Sjoberg2012} scans all possible splits, leading to an exponential scanning problem.
The advanced method in \citet{Sjoberg2012} uses a basis function expansion for the input dynamics and one for the inverse of the output dynamics. A scanning procedure over the basis functions is proposed as well. Compared to the brute-force method, the number of scans is lower, but the computational time can still be large.
The approach in \citet{Westwick2012} not only uses the \textsc{BLA}, but also the so-called quadratic \textsc{BLA} (\textsc{QBLA}), a higher-order \textsc{BLA} from the squared input to the output residual of the first-order \textsc{BLA}. By doing so, the number of possible splits is reduced significantly. Due to the higher-order nature of the \textsc{QBLA}, however, long measurements are needed to obtain an accurate estimate.
The nonparametric separation method proposed in \citet{Schoukens2014a} avoids the pole/zero assignment problem completely, but also uses the \textsc{QBLA}.

The method proposed in \citet{Schoukens2014b} and further developed in this paper uses again the first-order \textsc{BLA}. Using a well-designed excitation signal, the poles and the zeros of the input dynamics $R(q)$ shift with a frequency offset that can be chosen by the user, while the poles and the zeros of the output dynamics $S(q)$ remain invariant. Long measurement times can be avoided, because no use is made of higher-order \textsc{BLA}s.
This paper generalizes the basic ideas in \citet{Schoukens2014b} from cubic nonlinearities to polynomial nonlinearities. Moreover, experimental results on the Wiener-Hammerstein benchmark system \citep{Schoukens2009} are reported.

The rest of this paper is organized as follows.
The basic setup is described in \autoref{sec: Setup}.
A brief overview of the \textsc{BLA} is presented in \autoref{sec: BLA WH random-phase MS}.
The proposed method is presented in \autoref{sec: BLA WH phase-coupled MS}.
The experimental results on the Wiener-Hammerstein benchmark system are reported in \autoref{sec: Experimental results}.
Finally, the conclusions are drawn in \autoref{sec: Conclusion}.

\begin{figure}
	\centering
	\includegraphics[width=0.45\textwidth]{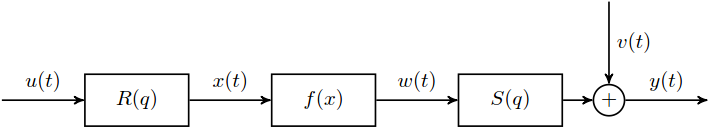}
	\caption{A Wiener-Hammerstein system
			($R$ and $S$ are linear dynamic systems and $f$ is a nonlinear static system).
			\label{fig: Wiener-Hammerstein system}
		}
\end{figure}

\section{Setup}
\label{sec: Setup}

This section introduces some notation. It also presents the considered Wiener-Hammerstein system and the assumptions.

\subsection{Notation}

Without loss of generality, discrete-time systems are considered. Hence, the integer $t$ denotes the time as a number of samples. The results in this paper generalize to continuous-time systems with some minor modifications.

\begin{notation}[$X(k)$ and $x(t)$]
\label{notation: (inverse) DFT}
	The discrete Fourier transform (\textsc{DFT}) of a time domain signal $x(t)$ is denoted by \mbox{$X(k) = X(e^{j \omega_k})$}, and is given by
\begin{equation}
	X(k) = \frac{1}{\sqrt{N}} \sum_{t=0}^{N-1} x(t) e^{-j 2\pi \frac{k}{N} t}
	\, .
\end{equation}
The inverse \textsc{DFT} is given by
\begin{equation}
\label{eq: inverse DFT}
	x(t) = \frac{1}{\sqrt{N}} \sum_{k=-N/2+1}^{N/2} X(k) e^{j 2\pi \frac{k}{N} t}
	\, .
\end{equation}
\end{notation}
\begin{notation}[$q^{-1}$]
	The backward shift operator is denoted by $q^{-1}$, i.e. \mbox{$q^{-1} x(t) = x(t - 1)$}.
\end{notation}
\begin{notation}[$O(\cdot)$]
	The notation $h$ is an $O(N^{\alpha})$ indicates that for $N$ big enough, \mbox{$\lvert h(N) \rvert \le c {N}^{\alpha}$}, where $c$ is a strictly positive real number.
\end{notation}
\begin{notation}[${(\cdot)}^{*}$]
	The complex conjugate of a complex number $X$ is denoted by ${X}^{*}$.
\end{notation}

\subsection{The Wiener-Hammerstein system}

Consider the Wiener-Hammerstein system in \autoref{fig: Wiener-Hammerstein system}, given by
\begin{equation}
	\label{eq: WH}
	\begin{aligned}
		x(t)	& = R(q) u(t)	\, ,\\
		w(t)	& = f(x(t))	\, ,\\
		y(t)		& = S(q) w(t) + v(t)	\, ,
	\end{aligned}
\end{equation}
where $R(q)$ and $S(q)$ are linear time-invariant (LTI) discrete-time transfer functions, i.e.
\begin{equation}
	\begin{aligned}
		R(q)		& = \frac{B_R(q)}{A_R(q)} 
								= \frac	{\sum_{l = 0}^{n_R} 
										b_{R,l} q^{-l}}
									{\sum_{l = 0}^{m_R} 
										a_{R,l} q^{-l}}	
								\, ,\\
		S(q)	& = \frac{B_S(q)}{A_S(q)} 
								= \frac	{\sum_{l = 0}^{n_S}
										b_{S,l} q^{-l}}
									{\sum_{l = 0}^{m_S}
										a_{S,l} q^{-l}}
								\, ,
	\end{aligned}
\end{equation}
and where $f(x)$ is a static nonlinear function.
Only the input $u(t)$ and the noise-corrupted output $y(t)$ are available for measurement.

\subsection{Assumptions}

This paper addresses the generation of initial estimates for the linear dynamics $R(q)$ and $S(q)$.
To do this, assumptions (A\ref{ass: polynomial nonlinearity})--(A\ref{ass: steady-state}) are made.

\begin{enumerate}
\renewcommand\labelenumi{(A\theenumi)}
	\item \label{ass: polynomial nonlinearity}
		The static nonlinearity $f(x)$ can be arbitrarily well approximated by a polynomial in the interval $[\min(x(t)), \max(x(t))]$.
		Hence, $f(x)$ can be thought of as \mbox{$f(x) = \sum_{D = 0}^\infty \gamma_D x^D$}.
	\assumptionintertext{Note that a uniformly convergent polynomial approximation of a continuous nonlinearity is always possible on a closed interval due to the Weierstrass approximation theorem \citep{WeissteinWeierstrass}. The type of convergence can be relaxed to mean-square convergence, thus allowing for some discontinuous nonlinearities as well.}
	\item \label{ass: nonlinear odd term}
		At least one nonlinear odd term is present, i.e. there exists an odd \mbox{$D \ge 3$} for which \mbox{$\gamma_D \ne 0$}.
	\assumptionintertext{Assumption (A\ref{ass: nonlinear odd term}) is slightly stricter than the assumption of non-evenness of the nonlinearity, which is typically made in other BLA approaches. When facing an even nonlinearity, a slight DC offset is typically added to the input, such that the nonlinearity is no longer even around the new set-point. If the nonlinearity only has polynomial terms up to second degree, i.e. it is a parabola, then adding a DC offset to the input can make the nonlinearity non-even, but it will not create a nonlinear odd term.}
	\item \label{ass: output noise}
		The additive output noise $v(t)$ is a sequence of zero-mean filtered white noise  that is independent of the excitation signal $u(t)$.
	\assumptionintertext{Under this assumption, the classical least-squares framework is known to result in consistent estimates of the \textsc{BLA} \citep{Pintelon2012}. To simplify the notation, and without loss of generality, the results in this paper are presented in the noise-free case.
Input and process noise are not considered here. In general, this would result in biased estimates. A more involved errors-in-variables approach (e.g. \citet{Mu2014}) is required to obtain unbiased estimates in this more general case.}
	\item \label{ass: steady-state}
		The system operates in steady-state.
\end{enumerate}

\section{The \textsc{BLA} of a Wiener-Hammerstein system using random-phase multisines}
\label{sec: BLA WH random-phase MS}

This section briefly reviews the \textsc{BLA} of a system. Explicit expressions for the output spectrum of the \textsc{BLA} of a Wiener-Hammerstein system, excited by a random-phase multisine, are provided.
These expressions will be convenient to set the ideas for the proposed method in \autoref{sec: BLA WH phase-coupled MS}, where phase-coupled multisines are proposed.

\subsection{Random-phase multisine excitation}

This paper considers multisine excitations.
\begin{definition}[multisine]
	\label{def: MS}
	A signal $u(t)$ is a multisine if
	\begin{equation}
		\label{eq: MS}
		u(t) = \sum_{k = -N/2 +1}^{N/2}
						U_k e^{j 2 \pi \frac{k}{N} t}
		\quad \text{for } t = 0, 1, \ldots, N - 1
		\, ,
	\end{equation}
	where the Fourier coefficients \mbox{$U_k = {U}^{*}_{- k} = \lvert U_k \rvert e^{j\phi_k}$} are either zero (no excitation present at frequency line $k$) or have a normalized amplitude \mbox{$\vert U_k \rvert = \frac{1}{\sqrt{N}} \check{U}\left( \frac{k}{N} \right)$}, where \mbox{$\check{U}\left( \frac{\omega}{2\pi} \right) \in \mathds{R}^{+}$} is a uniformly bounded function.
\end{definition}
\begin{definition}[random-phase multisine]
	\label{def: random-phase MS}
	A signal $u(t)$ is a random-phase multisine if it is a multisine (see Definition~\ref{def: MS}) where the phases $\phi_k$ are independently and identically distributed with the property \mbox{$E\left\{e^{j \phi_k}\right\} = 0$}.
\end{definition}

\subsection{The best linear approximation}

The \textsc{BLA} of a system is defined as the linear system whose output approximates the system's output best in mean-square sense around the operating point \citep{Pintelon2012}, i.e.
\begin{definition}[best linear approximation]
	\begin{equation}
		\label{eq: BLA}
		G_\mathrm{BLA}(k) := \argmin_{G(k)} E_u\left\{\lVert 
			\tilde{Y}(k) - G(k) \tilde{U}(k) \rVert^2\right\}
		\, ,
	\end{equation}
	with
	\begin{equation}
		\begin{cases}
			\tilde{u}(t) = u(t) - E\left\{u(t)\right\}		\\
			\tilde{y}(t) = y(t) - E\left\{y(t)\right\}
		\end{cases}
		\, ,
	\end{equation}
	where $G_\mathrm{BLA}$ is the frequency response function (\textsc{FRF}) of the \textsc{BLA}, and where the expectation in~\eqref{eq: BLA} is taken with respect to the random input $u$.
\end{definition}
\begin{remark}
	In the remainder of this paper, it is assumed that the mean values are removed from the signals when a \textsc{BLA} is calculated. The notations $u$ and $y$ will be used, instead of $\tilde{u}$ and $\tilde{y}$.
\end{remark}
It can be shown that
\begin{equation}
	\label{eq: BLA random}
	G_\mathrm{BLA}(k) = \frac 	{S_{YU}(k)}
							{S_{UU}(k)}
	\, ,
\end{equation}
where the expectation in the cross-power and auto-power spectra is again taken with respect to the random input $u$.
Note that for periodic excitations, \eqref{eq: BLA random} reduces to \citep{Schoukens2012}
\begin{equation}
	\label{eq: BLA periodic}
	G_\mathrm{BLA}(k) = E_u\left\{\frac{Y(k)}{U(k)}\right\}
	\, .
\end{equation}
It follows from \eqref{eq: BLA random} and Bussgang's theorem~\citep{Bussgang1952} that the \textsc{BLA} of the considered Wiener-Hammerstein system for Gaussian excitations is proportional to the product of the underlying dynamics.
This is summarized in the following theorem.
\begin{theorem}
	\label{theorem: BLA WH}
	The \textsc{BLA} of the Wiener-Hammerstein system in \eqref{eq: WH}, excited by Gaussian noise or by a random-phase multisine, is equal to
	\begin{equation}
		G_\mathrm{BLA}(k) = c_\mathrm{BLA} R(k) S(k) + O(N^{-1})
		\, ,
	\end{equation}
	where $c_\mathrm{BLA}$ is a constant that depends upon the odd nonlinearities in $f(x)$ and the power spectrum of the input signal.
\end{theorem}

\begin{proof}
	This is shown in \citet[pp.~ 85--86]{Pintelon2012}.
\end{proof}

The constant $c_\mathrm{BLA}$ is nonzero if $f(x)$ is non-even around the operating point.
Theorem~\ref{theorem: BLA WH} shows that it is easy to measure the product \mbox{$R(k) S(k)$} of a Wiener-Hammerstein system by measuring its \textsc{BLA} for an (asymptotically) Gaussian excitation.

\subsection{The output spectrum of the \textsc{BLA}}

An explicit expression for the output spectrum of the \textsc{BLA} of a Wiener-Hammerstein system, excited by a random-phase multisine, is derived in this subsection.
For the more complex case of a phase-coupled multisine, a similar derivation will be used (see Theorem~\ref{theorem: FRFs phase-coupled MS}).

Under Assumptions~(A\ref{ass: polynomial nonlinearity}) and~(A\ref{ass: steady-state}), the noise-free output spectrum of the Wiener-Hammerstein system in \eqref{eq: WH} is equal to
\begin{equation}
	Y(k) = \gamma_0 + \sum_{D = 1}^\infty 
								\gamma_D Y_D(k)
	\, ,
\end{equation}
where $Y_D(k)$ is the noise-free output spectrum of a Wiener-Hammerstein system that contains a pure $D$th-degree nonlinearity (\mbox{$f(x) = x^D$}).
The multiplication in the time domain corresponds to a convolution in the frequency domain, and thus (also keeping in mind the normalization factor in the inverse \textsc{DFT} in \eqref{eq: inverse DFT})
\begin{multline}
	\label{eq: output spectrum D}
	Y_D(k) \\= 
						\left( \frac{1}{\sqrt{N}} \right)^{D - 1}
						S(k) 
						\sum_{l_1, l_2, \ldots, l_D = -N/2 +1}^{N/2}
						\prod_{i = 1}^D 
						R(l_i) 
						U(l_i)
	\, ,
\end{multline}
such that \mbox{$\sum_{i = 1}^D l_i = k$}.
The only terms in $Y_D(k)$ that contribute to the \textsc{BLA} are those where the product \mbox{$\prod_{i=1}^D U(l_i)$} has a phase \mbox{$\phi_k = \angle U(k)$}. Terms that also depend on \mbox{$\phi_{l \ne k}$} will be eliminated in the expected value \mbox{$E_u\left\{Y(k) {U}^{*}(k)\right\}$} in~\eqref{eq: BLA random} or in the expected value in~\eqref{eq: BLA periodic}.
The contributing terms are those where one of the $l_i$s is equal to $k$, and where the other factors combine pairwise to \mbox{$X(l) X(-l) = \lvert X(l) \rvert^2$}. Note that this is only possible if $D$ is odd. Summing up all the terms in~\eqref{eq: output spectrum D} that contribute to the \textsc{BLA} results in
\begin{equation}
	\label{eq: BLA contributions random-phase MS}
	Y_{D, \mathrm{BLA}}(k) =
	\begin{aligned}[t]
		& D! S(k) R(k) U(k) \left( \frac{1}{N}\sum_{l = -N/2 +1}^{N/2} \lvert X(l) \rvert^2 \right)^{\frac{D-1}{2}}		\\
		& + O(N^{-1})
	\, .
	\end{aligned}
\end{equation}
For example, for \mbox{$f(x) = x^3$}, the \textsc{BLA} is equal to
\begin{equation}
	G_\mathrm{BLA}(k) = 6 S(k) R(k) \left( \frac{1}{N} \right.\sum_{l = -N/2 +1}^{N/2} \lvert X(l) \rvert^2 \left. \vphantom{\frac{1}{N}}\right) + O(N^{-1})
	\, .
\end{equation}
The error term $O(N^{-1})$ is due to the fact that there are six permutations of \mbox{$(k, l, -l)$} if \mbox{$k \ne l$}, while there are only three permutations of \mbox{$(k, k, -k)$}.

\section{The proposed method}
\label{sec: BLA WH phase-coupled MS}

In this section, we propose to use a special type of multisines, namely phase-coupled multisines.
After defining phase-coupled multisines, it is shown how the use of these signals makes it possible to separate the input and the output dynamics.
Finally, a practical issue when working with phase-coupled multisines is addressed.

\subsection{Phase-coupled multisine excitation}

Phase-coupled multisines are multisines where specifically selected pairs of frequency lines where excitation is present have the same phase.
Depending on whether excitation is present at both even and odd, or only at odd frequency lines, we are dealing with a full or an odd phase-coupled multisine, respectively.

\begin{definition}[full phase-coupled multisine]
	\label{def: full phase-coupled MS}
	A signal $u(t)$ is a phase-coupled multisine if it is a multisine (see Definition~\ref{def: MS}) where excitation is only present at frequency lines $k$ for which
	\begin{equation}
		\pm k \in \left\{\left(	\frac{d}{2} + d i, 
							\frac{d}{2} + d i + s\right)\right\}
		\quad \text{for } i = 0, 1, \ldots, i_\mathrm{max}
		\, ,
	\end{equation}
	where $d$ and $s$ are integers (more details below),
	and if each of the frequency couples gets assigned an independently and identically distributed random phase \mbox{$\phi_{\frac{d}{2} + d i} = \phi_{\frac{d}{2} + d i + s}$} with the property \mbox{$E\left\{e^{j \phi_{\frac{d}{2} + d i}}\right\} = 0$}.
	For a full phase-coupled multisine, the even integer \mbox{$d \ge 4$} determines the frequency resolution (excitation is present only every $d$th frequency line, starting from the frequency lines $\frac{d}{2}$ and \mbox{$\frac{d}{2} + s$}), and \mbox{$s = c_\mathrm{shift} d + 1$} will determine the shift of the poles and the zeros of $R(q)$ with respect to those of $S(q)$, where \mbox{$c_\mathrm{shift} > 0$} is an integer.
\end{definition}

It can be useful to only have excitation present at the odd frequency lines, just as with random-phase multisines \citep{Schoukens2005}.
The definition of the phase-coupled multisine then slightly changes.
\begin{definition}[odd phase-coupled multisine]
	\label{def: odd phase-coupled MS}
	A signal $u(t)$ is an odd phase-coupled multisine if it is a phase-coupled multisine (see Definition~\ref{def: full phase-coupled MS}) where \mbox{$\frac{d}{2} \ge 5$} is an odd integer, and where \mbox{$s = c_\mathrm{shift} d + 2$}.
\end{definition}

To simplify the notation in the rest of the paper, define
\begin{equation}
	\label{eq: excited line group one}
	m := \frac{d}{2} + d i
	\, .
\end{equation}

\begin{remark}
\label{remark: frequency resolution phase-coupled MS}
Although this is not really necessary, the requirement \mbox{$d \ge 4$} for a full phase-coupled multisine, and \mbox{$\frac{d}{2} \ge 5$} for an odd phase-coupled multisine makes sure that there is no excitation present at the frequency lines \mbox{$k = m - s$} and \mbox{$k = m + 2s$}. Hence, the output spectrum at these frequency lines will not be disturbed by linear contributions.
\end{remark}

\subsection{New terms in the output spectrum of the \textsc{BLA}}

Since $U(m)$ and $U(m+s)$ have the same phase in a phase-coupled multisine, also other terms than the ones reported in~\eqref{eq: BLA contributions random-phase MS} will contribute to the \textsc{BLA} at the frequency lines $m$ and $m+s$. Moreover, the output spectrum will contain terms that are proportional to $S(k)$ and shifted versions of $R(k)$ at some frequency lines where no excitation is present. This is explained below.

Let us take a look at the terms in $Y_D(k)$ in~\eqref{eq: output spectrum D} where the product \mbox{$\prod_{i=1}^D U(l_i)$} has a phase \mbox{$\angle U(m)$}.
These are the terms where one of the frequency lines $l_i$ is equal to $m$ or $m+s$ ($U(m+s)$ has the same phase as $U(m)$), and where the remaining factors in the product \mbox{$\left(\frac{1}{\sqrt{N}}\right)^{D-1}\prod_{i = 1}^D R(l_i)  U(l_i)$} combine pairwise to a constant.
The possible values for these \mbox{$\frac{D - 1}{2}$} (complex) constants are either
\begin{subequations}
	\begin{equation}
		c_0 = \frac{1}{N} \sum_{l = -N/2 +1}^{N/2} 
									X(l) X(-l)
		\, ,
	\end{equation}
	\begin{equation}
		c_{-s} = \frac{1}{N} \sum_{l = -N/2 +1}^{N/2} 
									X(l) X(-(l+s))
		\, ,
	\end{equation}
	or
	\begin{equation}
		c_{s} = \frac{1}{N} \sum_{l = -N/2 +1}^{N/2} 
									X(-l) X(l+s)
		\, .
	\end{equation}
\end{subequations}
Note that whenever a pair of factors combines to $c_{-s}$, a frequency shift $-s$ is introduced in the sum \mbox{$\sum_{i = 1}^D l_i = k$}. Likewise, a frequency shift $s$ is introduced whenever a pair of factors combines to \mbox{$c_{s} = {(c_{-s})}^{*}$}.
There are thus \mbox{$D+1$} frequency lines $k$ that range from \mbox{$m - \frac{D-1}{2}s$} to \mbox{$m + \frac{D+1}{2}s$} in steps of $s$ where the product \mbox{$\prod_{i=1}^D U(l_i)$} has a phase \mbox{$\angle U(m)$}.
The smallest frequency line (\mbox{$k = m - \frac{D-1}{2}s$}) is obtained when one frequency line $l_i$ is equal to $m$, and when all the other frequency lines form pairs \mbox{$(l, -(l+s))$}.
The largest frequency line (\mbox{$k = m + \frac{D+1}{2}s$}) is obtained when one frequency line $l_i$ is equal to \mbox{$m + s$}, and when all the other frequency lines form pairs \mbox{$(-l, l+s)$}.
From the discussion above, the following theorem follows:
\begin{theorem}
\label{theorem: FRFs phase-coupled MS}
	Under Assumptions~(A\ref{ass: polynomial nonlinearity}) and~(A\ref{ass: steady-state}), and for \mbox{$i = -\frac{D - 1}{2}, -\frac{D - 1}{2} + 1, \ldots, \frac{D + 1}{2}$}, the expectation $E_u\left\{\frac{Y_D(m + i s)}{U(m)}\right\}$ is equal to
	\begin{equation}
		\begin{aligned}[t]
			& D! S(m + i s) \left[ R(m) \sum_{\sum_k s_k = i s} \prod_{k = 1}^{\frac{D-1}{2}} c_{s_k} \right.						\\
			& + \left. R(m+s) \frac{\lvert U(m+s) \rvert}{\lvert U(m) \rvert} \sum_{\sum_k s_k = (i-1) s} \prod_{k = 1}^{\frac{D-1}{2}} c_{s_k} \right]	\\
			& + O(N^{-1})
		\, ,
		\end{aligned}
	\end{equation}
	and \mbox{$E_u\left\{\frac{Y_D(-(m + i s))}{U(-m)}\right\} = {\left( E_u\left\{\frac{Y_D(m + i s)}{U(m)}\right\} \right)}^{*}$} is equal to
	\begin{equation}
	\label{eq: shifted BLA negative frequencies}
		\begin{aligned}[t]
			& D! S(-(m + i s)) \left[ R(-m) \sum_{\sum_k s_k = - i s} \prod_{k = 1}^{\frac{D-1}{2}} c_{s_k} \right.							\\
			& + \left. R(-(m+s)) \frac{\lvert U(m+s) \rvert}{\lvert U(m) \rvert} \sum_{\sum_k s_k = - (i-1) s} \prod_{k = 1}^{\frac{D-1}{2}} c_{s_k} \right]	\\
			& + O(N^{-1})
		\end{aligned}
	\end{equation}
	for the Wiener-Hammerstein system in~\eqref{eq: WH}, excited by a phase-coupled multisine (see Definitions~\ref{def: full phase-coupled MS} and~\ref{def: odd phase-coupled MS}), and where each \mbox{$s_k \in \{ -s, 0, s \}$}. 
\end{theorem}
\begin{proof}
	The theorem follows immediately from the discussion above.
\end{proof}
From here on, the error term $O(N^{-1})$ will be dropped.
For example, for \mbox{$D = 3$}, there are four frequency lines where \mbox{$E_u\left\{\frac{Y_3(k)}{U(m)}\right\}$} has a nonzero mean. These are listed below.
\begin{subequations}
\label{eq: shifted BLA cubic}
	\begin{enumerate}
		\item At frequency line $k = m - s$:
			\begin{equation}
				E_u\left\{\frac{Y_3(k)}{U(m)}\right\} = 6 S(m - s) R(m) c_{-s}
			\end{equation}
		\item At frequency line $k = m$:
			\begin{equation}
				E_u\left\{\frac{Y_3(k)}{U(m)}\right\} = 
				\begin{aligned}[t]
					& 6 S(m) R(m) c_0 								\\
					& + 6 S(m) R(m+s) \frac{\lvert U(m+s) \rvert}{\lvert U(m) \rvert} c_{-s}
				\end{aligned}
			\end{equation}
		\item At frequency line $k = m + s$:
			\begin{equation}
				E_u\left\{\frac{Y_3(k)}{U(m)}\right\} = 
				\begin{aligned}[t]
					&  6 S(m + s) R(m+s) \frac{\lvert U(m+s) \rvert}{\lvert U(m) \rvert} c_0	\\
					& + 6 S(m + s) R(m) c_s
				\end{aligned}
			\end{equation}
		\item At frequency line $k = m + 2s$:
			\begin{equation}
				E_u\left\{\frac{Y_3(k)}{U(m)}\right\} = 6 S(m + 2s) R(m+s) \frac{\lvert U(m+s) \rvert}{\lvert U(m) \rvert} c_s
			\end{equation}
	\end{enumerate}
\end{subequations}
The results at the frequency lines \mbox{$k = m - s$} and \mbox{$k = m + 2 s$} are of particular interest. At these frequency lines, \mbox{$E_u\left\{\frac{Y_3(k)}{U(m)}\right\}$} is proportional to \mbox{$S(k) R(k + s)$} and \mbox{$S(k) R(k - s)$}, respectively. The frequency shift of the input dynamics $R$ over a frequency $-s$ (or $s$) creates a shift of the poles and the zeros of $R$ over a frequency $-s$ (or $s$). Note that the shifted poles and zeros are no longer real nor paired in complex conjugated couples. Hence, the rational transfer functions related to \mbox{$E_u\left\{\frac{Y_3(m - s)}{U(m)}\right\}$} and \mbox{$E_u\left\{\frac{Y_3(m + 2 s)}{U(m)}\right\}$} will have complex coefficients instead of real coefficients. This will be used in the next subsection.

\subsection{The shifted \textsc{BLA} and parametric smoothing}

For simplicity, we continue to explain the idea on a third-degree nonlinearity. The generalization to an arbitrary degree $D$ is explained in \autoref{sec: generalization to arbitrary degree}.
First, the \textsc{FRF}~measurements \mbox{$E_u\left\{\frac{Y_3(k)}{U(m)}\right\}$} and \mbox{$E_u\left\{\frac{Y_3(k)}{U(-m)}\right\}$} will be collected at appropriate frequency lines $k$, such that only contributions are selected where the input and output dynamics shift over a unique frequency offset.
Next, a parametric transfer function model will be identified to get direct access to the poles and the zeros of the system. As the shifted poles result in a transfer function model with complex coefficients, an adapted frequency domain estimator \citep{Peeters2001} will be used.
Finally, the input and the output dynamics will be split by separating the shifting poles and zeros from those that do not move.

From~\eqref{eq: shifted BLA negative frequencies}, it follows that
\begin{equation}
	E_u\left\{\frac{Y_3(-(m-s))}{U(-m)}\right\} = 6 S(-(m-s)) R(-m) c_s
\end{equation}
and that
\begin{equation}
	E_u\left\{\frac{Y_3(-(m+2s))}{U(-m)}\right\} = 
	\begin{aligned}[t]
		& 6 S(-(m+2s)) R(-(m+s)) 						\\
		& \times \frac{\lvert U(m+s) \rvert}{\lvert U(m) \rvert} c_{-s}
	\end{aligned}
\end{equation}
Hence, \mbox{$E_u\left\{\frac{Y_3(k)}{U(-m)}\right\}$} is proportional to \mbox{$S(k) R(k - s)$} at the frequency lines \mbox{$k = -(m - s)$}, while it is proportional to \mbox{$S(k) R(k + s)$} at the frequency lines \mbox{$k = -(m + 2s)$}.
Therefore, by analogy with~\eqref{eq: BLA periodic}, define the shifted \textsc{BLA}s $G_\mathrm{SBLA}^{+}(k)$ and $G_\mathrm{SBLA}^{-}(k)$ by collecting \mbox{$E_u\left\{\frac{Y(k)}{U(m)}\right\}$} and \mbox{$E_u\left\{\frac{Y(k)}{U(-m)}\right\}$} at the appropriate frequency lines $k$, such that contributions proportional to \mbox{$S(k) R(k + s)$} and \mbox{$S(k) R(k - s)$}, respectively, result:
\begin{definition}[shifted best linear approximation]
	\label{def: shifted BLA}
	For a system excited by a phase-coupled multisine (see Definitions~\ref{def: full phase-coupled MS} and~\ref{def: odd phase-coupled MS}), the shifted \textsc{BLA} $G_\mathrm{SBLA}^{+}(k)$ is defined as
	\begin{equation}
		G_\mathrm{SBLA}^{+}(k) :=
		\begin{cases}
			E_u\left\{\frac{Y(k)}{U(m)}\right\}	& \text{at $k = m-s$} 	\\
			E_u\left\{\frac{Y(k)}{U(-m)}\right\}	& \text{at $k = -(m+2s)$}
		\end{cases}
	\end{equation}
	while the shifted \textsc{BLA} $G_\mathrm{SBLA}^{-}(k)$ is defined as
	\begin{equation}
		G_\mathrm{SBLA}^{-}(k) :=
		\begin{cases}
			E_u\left\{\frac{Y(k)}{U(m)}\right\}	& \text{at $k = m+2s$} 	\\
			E_u\left\{\frac{Y(k)}{U(-m)}\right\}	& \text{at $k = -(m-s)$}
		\end{cases}
	\end{equation}
	with $m$ defined in~\eqref{eq: excited line group one}.
\end{definition}
Since \mbox{$G_\mathrm{SBLA}^{-}(k) = {(G_\mathrm{SBLA}^{+}(-k))}^{*}$}, we can focus completely on one of both, e.g. $G_\mathrm{SBLA}^{-}(k)$.

Next, a parametric transfer function model is identified on $G_\mathrm{SBLA}^{-}(k)$, using a weighted least-squares estimator \citep{Peeters2001}
\begin{subequations}
	\label{eq: parametric modeling shifted BLA}
	\begin{equation}
		\hat{\theta} = \argmin_{\theta} K(\theta)
		\, ,
	\end{equation}
	where the cost function $K(\theta)$ is equal to
	\begin{equation}
		\label{eq: cost function shifted BLA}
		\frac{1}{N} \sum_{k \in \{ m+2s, -(m-s) \}} \frac{\lvert G_\mathrm{SBLA}^{-}(k) - G_\mathrm{SBLA}^{-}(k, \theta) \rvert^2}{\sigma_{G_\mathrm{SBLA}^{-}}^2(k)}
		\, .
	\end{equation}
\end{subequations}
Here, \mbox{$G_\mathrm{SBLA}^{-}(k, \theta)$} is a parametric transfer function model with the complex coefficients $b_\mathrm{SBLA,l}$ and $a_\mathrm{SBLA,l}$ collected in the parameter vector $\theta$:
\begin{equation}
	G_\mathrm{SBLA}^{-}(k, \theta) = \frac 	{\sum_{l = 0}^{n_R + n_S} 
						b_{\mathrm{SBLA},l} e^{-j 2\pi \frac{k}{N}l}}
						{\sum_{l = 0}^{m_R + m_S} 
						a_{\mathrm{SBLA},l} e^{-j 2\pi \frac{k}{N}l}}
	\, ,
\end{equation}
and \mbox{$\sigma_{G_\mathrm{SBLA}^{-}}^2(k)$} is the sample variance of \mbox{$G_\mathrm{SBLA}^{-}(k)$}.
Since \mbox{$G_\mathrm{SBLA}^{-}(k)$} is proportional to \mbox{$S(k)R(k-s)$}, the poles (and the zeros) of $R$ will be shifted over a frequency $s$, and will no longer be real, nor complex conjugated. By comparing the poles (and the zeros) of \mbox{$G_\mathrm{SBLA}^{-}(k, \theta)$} and their complex conjugates, the poles (and the zeros) of $R$ and $S$ can be separated. Indeed, a frequency shift $e^{j 2 \pi \frac{2 s}{N}}$ will be visible in the complex plane between the shifted poles (and zeros) of $R$ and their complex conjugates, while the poles (and zeros) of $S$ will not show this shift. \autoref{fig: pole shifting} shows an example of the pole shifting.

This approach has the potential to discriminate between a Wiener, a Hammerstein, and a Wiener-Hammerstein model, based on whether there are shifting poles/zeros (and thus input dynamics), fixed poles/zeros (and thus output dynamics), or both. A more detailed discussion on structure discrimination is provided in \citet{Schoukens2015}.

\begin{figure}
	\centering
	\includegraphics[width=0.45\textwidth]{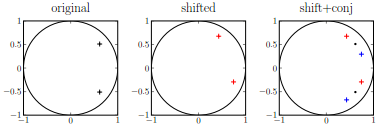}
	\caption{Example of shifting poles
			(Left: original poles; Middle: shifted poles; Right: Shifted and conjugated poles).
			\label{fig: pole shifting}
		}
\end{figure}

\subsection{Generalization to an arbitrary degree $D$}
\label{sec: generalization to arbitrary degree}

The basic idea, that was explained for a cubic nonlinearity in the previous subsection, is now generalized to an arbitrary degree $D$.
Under Assumption~(A\ref{ass: nonlinear odd term}), the result in Theorem~\ref{theorem: FRFs phase-coupled MS} can be used to separate the input and the output dynamics.
The main concern, however, is that the shifted \textsc{BLA}s will contain contributions where the input and the output dynamics are shifted over distinct frequencies if \mbox{$D \ge 5$}. This will be addressed now.

The results in Theorem~\ref{theorem: FRFs phase-coupled MS} show that the shifted \textsc{BLA} $G_\mathrm{SBLA}^{-}(k)$ will not only contain terms that are proportional to \mbox{$S(k) R(k-s)$} for \mbox{$D \ge 5$}, but also terms that are proportional to \mbox{$S(k) R(k-2s)$}. One way to deal with this is to redefine the shifted \textsc{BLA}, thereby only considering the outer frequency lines \mbox{$k = m - \frac{D_\mathrm{max}-1}{2}s$} and \mbox{$k = m + \frac{D_\mathrm{max}+1}{2}s$}, where a unique frequency shift of $R$ is present. The main disadvantage of this approach is that it requires knowledge of the maximal degree of nonlinearity $D_\mathrm{max}$. Moreover, in order not to be disturbed by linear contributions at these frequency lines, the frequency resolution of the phase-coupled multisine excitation should be lowered (cfr. Remark~\ref{remark: frequency resolution phase-coupled MS}).
Therefore, a different approach is followed here.

Both the contributions proportional to \mbox{$S(k) R(k-s)$} and \mbox{$S(k) R(k-2s)$} will make the poles of $R$ shift over a positive frequency offset, while those of $S$ will remain fixed. Hence, the poles and the zeros of $R$ and $S$ can still be separated.
Since two distinct frequency shifts of the input dynamics are present in the shifted \textsc{BLA}, its parametric model should have a larger model order, namely the model order of $S$ plus two times the model order of $R$. This will not be done, however, because the contributions that are proportional to \mbox{$S(k) R(k-s)$} are dominant over those that are proportional to \mbox{$S(k) R(k-2s)$} (see \autoref{app: Dominant terms}), certainly if we assume that the cubic nonlinearities dominate the higher-order contributions. Hence, the terms in $G_\mathrm{SBLA}^{-}(k)$ that are not proportional to \mbox{$S(k) R(k-s)$} will simply be considered as nuisance terms.

\subsection{Time origin}

Since phase-coupled multisines rely on the fact that some frequency lines have equal phase, the time origin is important.
This synchronization issue is addressed now.

Suppose that the input and the output measurements are shifted over a time \mbox{$\delta$}, i.e. $t$ becomes \mbox{$t - \delta$}. This results in a frequency-dependent phase shift in the input and the output spectrum. For example, $U(k)$ becomes \mbox{$U(k) e^{j k \Delta}$}, with \mbox{$\Delta = -\frac{2 \pi \delta}{N}$}. Since the shifted \textsc{BLA} is defined as the expected value of the ratio of the input and the output spectrum at distinct frequency lines (see Definition~\ref{def: shifted BLA}), a phase shift is present in this shifted \textsc{BLA}:
\begin{equation}
	E_u\left\{\frac{Y(k)}{U(m)}\right\} \rightarrow E_u\left\{\frac{Y(k)}{U(m)} e^{j (k-m) \Delta}\right\}
	\, .
\end{equation}

Since \mbox{$\delta$} does not need to be an integer, the compensation for a shifted time origin will be done in the frequency domain.
The phase shift $\Delta$ can be determined for each pair of frequency lines at which excitation is present as
\begin{equation}
	\Delta(m) = \frac{\angle U(m + s) - \angle U(m)}{s}
	\, ,
\end{equation}
and the expected values in the shifted \textsc{BLA}s can be compensated by multiplying \mbox{$E_u\left\{\frac{Y(k)}{U(m)}\right\}$} with \mbox{$e^{j (m-k) \Delta(m)}$}, and \mbox{$E_u\left\{\frac{Y(k)}{U(-m)}\right\}$} with \mbox{$e^{j (-m-k) \Delta(m)}$}.

\section{Experimental results}
\label{sec: Experimental results}

This section illustrates the proposed method on experimental data obtained from the Wiener-Hammerstein benchmark system \citep{Schoukens2009}.

\subsection{Device}
\label{sec: description of the device}

The Wiener-Hammerstein benchmark system is an electronic circuit with a Wiener-Hammerstein structure. It consists of a diode-resistor network sandwiched in between two third-order filters (see \autoref{fig: benchmark system}). The input filter $R$ is a Chebyshev low-pass filter with a ripple of \SI{0.5}{\deci\bel} and a cut-off frequency of \SI{4.4}{\kilo\hertz}. The output filter $S$ is an inverse Chebyshev filter with a stop-band attenuation of \SI{40}{\deci\bel} starting at \SI{5}{\kilo\hertz}. It has a transmission zero in the frequency band of interest.

\begin{figure}
	\centering
	\includegraphics[width=0.45\textwidth]{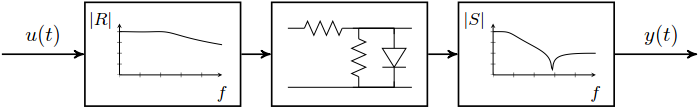}
	\caption{The Wiener-Hammerstein benchmark system consists of a diode-resistor network sandwiched in between two third-order Chebyshev filters.
			\label{fig: benchmark system}
		}
\end{figure}

\subsection{Measurement data}

The benchmark data were obtained using a filtered Gaussian noise excitation \citep{Schoukens2009} and are therefore not used in this paper. Here, a random-phase multisine is used to measure the \textsc{BLA}, and an odd phase-coupled multisine is used to measure the shifted \textsc{BLA}. For both excitations, the input and output are measured at a sample frequency of \SI{78125}{\hertz}.

The random-phase multisine contains \mbox{$N = 8192$} samples. It has a flat amplitude spectrum with $682$~frequencies ranging from \SI{19}{\hertz} to \SI{13800}{\hertz} where excitation is present, and an rms level of \SI{380}{\milli\volt}. This relatively small rms level is chosen to keep the nonlinear distortion level small. Seven phase realizations and three periods are applied. The first period is removed to avoid the effects of transients.

The phase-coupled multisine contains \mbox{$N = 8192$} samples as well, and also has a flat amplitude spectrum. With \mbox{$d = 10$}, \mbox{$s = 242$}, and \mbox{$i_\mathrm{max} = 111$}, there are $224$~frequencies that range from \SI{47}{\hertz} to \SI{12941}{\hertz} where excitation is present. The signal is normalized to have a maximal amplitude of \SI{2}{\volt}. This amplitude level corresponds more or less to that of the Wiener-Hammerstein benchmark data. Again, three periods are applied where the first one is removed to avoid the effects of transients. One thousand phase realizations are applied to almost completely average out the nonlinear distortions in the shifted \textsc{BLA} so that we can show a high-quality nonparametric estimate of the shifted \textsc{BLA} (see \autoref{fig: shifted BLA WH}). However, much less phase realizations can be used, since the parametric modeling step will also reduce the impact of the nonlinear distortions (and the noise). For example, ten phase realizations can be enough (see \autoref{sec: discussion results}).

\begin{figure}
	\centering
	\includegraphics[width=0.4\textwidth]{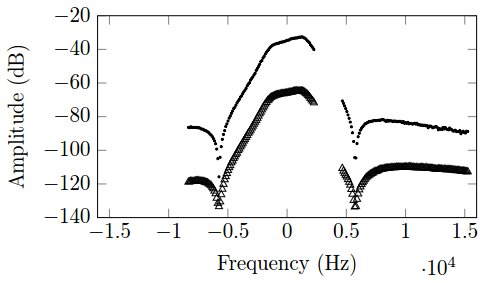}
	\caption{The shifted \textsc{BLA} (black dots) is not symmetric around the origin. Its standard deviation is shown in black triangles.
			\label{fig: shifted BLA WH}
		}
\end{figure}

\subsection{The \textsc{BLA} and the shifted \textsc{BLA}}
\label{sec: BLA and shifted BLA}

The \textsc{BLA} is first estimated non-parametrically via the robust method \citep{Schoukens2012} using the random-phase multisine data. Next, a parametric fourth-order transfer function model is estimated on top of the non-parametric \textsc{BLA} by minimizing a weighted least-squares cost function that is similar to~\eqref{eq: cost function shifted BLA}, but with real parameters $\theta$. Although a sixth-order model is expected ($R$ and $S$ are both third-order filters), a fourth-order model seems enough for the data at hand (see \autoref{fig: BLA WH}).
The poles and the zeros of the parametric model are used as a reference to compare them with those of the shifted \textsc{BLA}.

\begin{figure}
	\centering
	\includegraphics[width=0.4\textwidth]{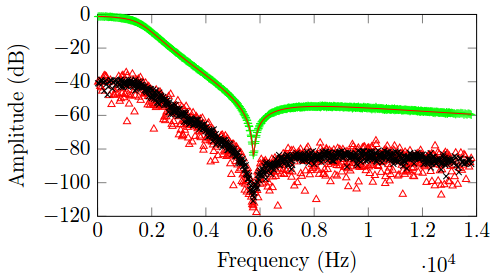}
	\caption{A parametric fourth-order model (red line) can explain the non-parametric \textsc{BLA} (green) (the magnitude of the complex difference between the \textsc{BLA}s (red triangles) coincides with the estimated total distortion level (black)).
			\label{fig: BLA WH}
		}
\end{figure}

The shifted \textsc{BLA} is estimated non-parametrically by averaging over the \num{1000} realizations of the input. The averaged value is shown together with its standard deviation in \autoref{fig: shifted BLA WH}. Observe that the amplitude characteristic is not symmetric around the origin. Hence, a parametric transfer function model with complex coefficients is required. A fourth-order model is estimated using the weighted least-squares approach in~\eqref{eq: parametric modeling shifted BLA}. The poles of this model, together with their complex conjugates and the original poles of the \textsc{BLA} are shown in \autoref{fig: poles WH}. A similar picture for the zeros is shown in \autoref{fig: zeros WH}.

\begin{figure}
	\centering
	\includegraphics[width=0.4\textwidth]{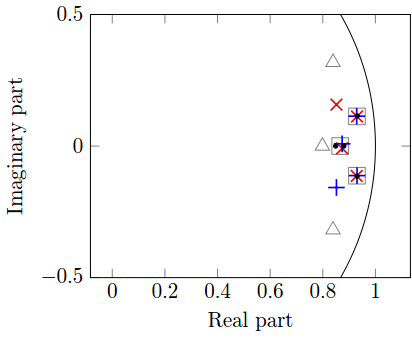}
	\caption{The poles of the shifted \textsc{BLA} are shown in red crosses, their complex conjugates in blue pluses, and the original poles of the \textsc{BLA} in black dots. One pole (where red and blue are not on top of each other) shows a clear shift and can be assigned to the input filter $R$. The three other poles (almost) remain invariant and can be assigned to the output filter $S$. This is in good agreement with the internal structure of the filters in \autoref{fig: benchmark system}. Their poles are added in gray as a reference (triangles: $R$, squares: $S$).
			\label{fig: poles WH}
		}
\end{figure}

\begin{figure}
	\centering
	\includegraphics[width=0.4\textwidth]{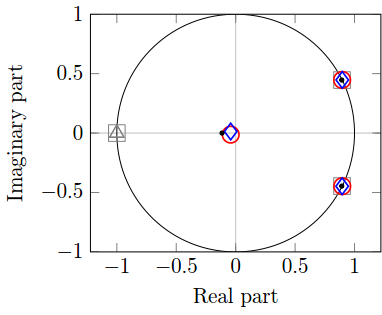}
	\caption{Zeros of the shifted \textsc{BLA} are shown in red circles, their complex conjugates in blue diamonds, and zeros of the \textsc{BLA} in black dots. One zero outside the unit circle is not shown. The complex pair of zeros can clearly be assigned to the output filter $S$. The other zeros cannot be classified. The zeros of the filters in \autoref{fig: benchmark system} are added in gray as a reference (triangles: $R$, squares: $S$).
			\label{fig: zeros WH}
		}
\end{figure}

\subsection{Discussion of the results}
\label{sec: discussion results}

One real pole can be assigned to the input filter $R$, since the corresponding pole of the shifted \textsc{BLA} shows a clear shift of about \SI{21}{\degree} with respect to its complex conjugate. This shift also nicely corresponds to the expected shift of \mbox{$\frac{2 s}{N}\SI{360}{\degree} = \SI{21.3}{\degree}$}. The other poles do (almost) not move. They can be assigned to the output filter $S$. Considering the internal structures of the filters (see \autoref{sec: description of the device}), the poles are assigned correctly (see also \autoref{fig: poles WH}).
The input filter $R$ should have a complex conjugate pole pair as well, but its effect on the FRFs seems unnoticeable due to the presence of the transmission zero in $S$ (see also \autoref{sec: BLA and shifted BLA}).

The complex pair of zeros can be assigned to the output filter $S$, since these zeros do not shift. This pair of zeros is assigned correctly, as it corresponds to the transmission zero in the output filter (see \autoref{sec: description of the device} and \autoref{fig: zeros WH}). The real zero in \autoref{fig: zeros WH} cannot be clearly assigned. Although a shift of about \SI{35}{\degree} is present between the corresponding zero of the shifted \textsc{BLA} and its complex conjugate, its amplitude is so small that a small uncertainty on the zero position can completely change its classification.
The inability to assign most of the zeros is to be expected, since, except for the transmission zero, the true zeros are at $-1$, and there is no excitation present in that part of the frequency band. Therefore, the uncertainty on these estimated zero positions is large, which prohibits their classification.

Similar results can be obtained when only ten phase realizations are used instead of one thousand. The \num{1000} experiments were split in \num{100} groups of ten phase realizations. In \num{61} cases, a correct assignment of the poles and zeros could clearly be made. In \num{36} of the remaining cases, the most damped real pole could have been wrongfully assigned to $S$, while the least damped real pole would then have been wrongfully assigned to $R$. The wrong assignment of these poles is likely to have a small impact on the further initialization of the Wiener-Hammerstein model, since both real poles lie close to each other. Only in \num{3} cases, it was unclear whether both real poles should be assigned to $R$, to $S$, or one to $R$ and one to $S$.

\section{Conclusion}
\label{sec: Conclusion}

The best linear approximation (\textsc{BLA}) of a Wiener-Hammerstein system that is excited by Gaussian noise or a random-phase multisine is proportional to the product of the underlying linear dynamic blocks. By applying a more specialized phase-coupled multisine, it is shown that a modified \textsc{BLA} expression is proportional to the product of the output dynamics and a frequency-shifted version of the input dynamics. On the basis of the non-parametric measurement of this modified \textsc{BLA}, it is shown to be possible to assign the identified poles and zeros to either the input or output dynamics, provided that the poles and zeros are properly excited. This is confirmed by experimental results on the Wiener-Hammerstein benchmark system.

The proposed method has the potential to discriminate between a Wiener, a Hammerstein, and a Wiener-Hammerstein model, based on whether there are shifting poles/zeros (and thus input dynamics), fixed poles/zeros (and thus output dynamics), or both.

Future work is to derive variance expressions for the shifted BLA, from which bounds on the estimated shifted poles/zeros can be calculated. With these bounds, it will be possible to determine whether a pole/zero significantly shifted or not. This will enable an automatic assignment of the poles and the zeros.

\begin{ack}
The research leading to these results has received funding from the European Research Council under the European Union's Seventh Framework Programme (FP7/2007-2013)/ERC Grant Agreement n.~320378.
This work was also supported in part by the Fund for Scientific Research (FWO-Vlaanderen), by the Flemish Government (Methusalem 1), and by the Belgian Federal Government (IAP DYSCO VII/19).
\end{ack}

\bibliographystyle{model5-names}
\bibliography{References_WH_pcMS}

\appendix
\section{Dominant terms in the shifted \textsc{BLA}}
\label{app: Dominant terms}

The contributions in \mbox{$G_\mathrm{SBLA}^{-}(k)$} that are proportional to \mbox{$S(k) R(k-s)$} are dominant over those that are proportional to \mbox{$S(k) R(k-2s)$}. This is shown in this appendix.

First, it will be shown that \mbox{$c_0 > \lvert c_s \rvert$}.
We have that
\begin{equation}
	c_0 = \frac{2}{N} \sum_m \lvert X(m) \rvert^2 + \lvert X(m+s) \rvert^2
	\, ,
\end{equation}
and that
\begin{equation}
	\lvert c_s \rvert
	\begin{aligned}[t]
		& = \frac{2}{N} \left\lvert \sum_m X(-m) X(m + s) \right\rvert	\\
		& \le \frac{2}{N} \sum_m \lvert X(-m) X(m + s) \rvert
	\, .
	\end{aligned}
\end{equation}
By working out \mbox{$(\lvert X(m) \rvert - \lvert X(m+s) \rvert)^2$} and rearranging the terms, we have
\begin{multline}
	\lvert X(m) \rvert^2 + \lvert X(m+s) \rvert^2 \\ = 2 \lvert X(m) \rvert \lvert X(m+s) \rvert + (\lvert X(m) \rvert - \lvert X(m+s) \rvert)^2
	\, ,
\end{multline}
and thus \mbox{$c_0 \ge 2 \lvert c_s \rvert$}.

Let $\alpha_D(k)$ be the ratio of the contributions in $Y_D(k)$ that are proportional to \mbox{$S(k) R(k-s)$} and those that are proportional to \mbox{$S(k) R(k-2s)$}. Then we need to show that \mbox{$\lvert \alpha_D(k) \rvert > 1$}. From Theorem~\ref{theorem: FRFs phase-coupled MS}, we have
\begin{equation}
	\alpha_D(k) = \frac 	{ 		\frac{\lvert U(m+s) \rvert}{\lvert U(m) \rvert} 	\sum\limits_{\sum_k s_k = s} \prod_{k = 1}^{\frac{D-1}{2}} c_{s_k}}
				{\phantom{ 	\frac{\lvert U(m+s) \rvert}{\lvert U(m) \rvert}}	\sum\limits_{\sum_k s_k = 2 s} \prod_{k = 1}^{\frac{D-1}{2}} c_{s_k}}
\end{equation}
at the frequency lines \mbox{$k = m + 2 s$}, and
\begin{equation}
	\alpha_D(k) = \frac 	{\phantom{ 	\frac{\lvert U(m+s) \rvert}{\lvert U(m) \rvert}}	\sum\limits_{\sum_k s_k = s} \prod_{k = 1}^{\frac{D-1}{2}} c_{s_k}}
				{ 		\frac{\lvert U(m+s) \rvert}{\lvert U(m) \rvert}	\sum\limits_{\sum_k s_k = 2 s} \prod_{k = 1}^{\frac{D-1}{2}} c_{s_k}}
\end{equation}
at the frequency lines \mbox{$k = -(m - s)$}.
As a factor \mbox{$\frac{\lvert U(m+s) \rvert}{\lvert U(m) \rvert} > 1$} would be advantageous in one case, and disadvantageous in the other case, simply consider \mbox{$\frac{\lvert U(m+s) \rvert}{\lvert U(m) \rvert} = 1$}.
Since the $s_k$s in the numerator should sum up to $s$, at least one factor $c_s$ should be present in the numerator. Likewise, two factors $c_s$ should be present in the denominator. The remaining $s_k$s should sum up to zero. Hence,
\begin{equation}
	\alpha_D(k) = \frac{c_s}{c_s c_s}	\frac 	{\sum\limits_{\sum_k s_k = 0} \prod_{k = 1}^{\frac{D-3}{2}} c_{s_k}}
							{\sum\limits_{\sum_k s_k = 0} \prod_{k = 1}^{\frac{D-5}{2}} c_{s_k}}
	\, .
\end{equation}
A zero-sum of the $s_k$s is possible if all of them are zero. Furthermore, one or more pairs \mbox{$c_0 c_0$} can be replaced by \mbox{$c_s c_{-s} = \lvert c_s \rvert^2$}. Hence,
\begin{equation}
	\lvert \alpha_D(k) \rvert 
	\begin{aligned}[t]
		& =	\frac{1}{\lvert c_s \rvert}	\frac 	{\sum\limits_{i=0}^{\floor*{\frac{D-3}{4}}} c_0^{\frac{D-3}{2} - 2 i} \lvert c_s \rvert^{2 i}}
								{\sum\limits_{i=0}^{\floor*{\frac{D-5}{4}}} c_0^{\frac{D-5}{2} - 2 i} \lvert c_s \rvert^{2 i}}	\\
		& \ge \frac{c_0}{\lvert c_s \rvert}
	\, ,
	\end{aligned}
\end{equation}
and since \mbox{$c_0 \,\ge 2 \lvert c_s \rvert$}, we have \mbox{$\lvert \alpha_D(k) \rvert \,\ge 2 > 1$}.
For \mbox{$D=3$}, there are no contributions proportional to \mbox{$S(k) R(k-2s)$}, so that \mbox{$\alpha_3(k) \rightarrow \infty$}.

\end{document}